\documentclass[11pt,english]{article}
\usepackage[latin9]{inputenc}
\usepackage{geometry}
\usepackage{verbatim}
\usepackage{float}
\usepackage{amsthm}
\usepackage{amsmath}
\usepackage{amssymb}
\usepackage{esint}
\usepackage{mathtools}
\usepackage[colorlinks,citecolor=blue]{hyperref}

\newif\ifFULL
\FULLtrue

\makeatletter

\floatstyle{ruled}
\newfloat{algorithm}{tbp}{loa}
\providecommand{\algorithmname}{Algorithm}
\floatname{algorithm}{\protect\algorithmname}

  \newtheorem{theorem}{Theorem}
  \numberwithin{theorem}{section}
  \numberwithin{equation}{section}
  \theoremstyle{definition}
  \newtheorem{defn}[theorem]{\protect\definitionname}
 \theoremstyle{definition}
  \newtheorem{example}[theorem]{Example}

  \theoremstyle{remark}
  \newtheorem{claim}[theorem]{Claim}
  \newtheorem{observation}{Observation}

\theoremstyle{plain}
\newtheorem{thm}[theorem]{\protect\theoremname}
  \theoremstyle{plain}
  \newtheorem{lem}[theorem]{\protect\lemmaname}
  \theoremstyle{plain}
  \newtheorem{prop}[theorem]{\protect\propositionname}

\usepackage{complexity,color,bm}
\usepackage{colortbl}

\providecommand{\E}{\mathrm{E}}

\newcommand{\myvec}{\mathbf}
\newcommand{\x}{\myvec{x}}
\newcommand{\X}{\myvec{X}}

\newcommand{\Y}{\myvec{Y}}

\newcommand{\ignore}[1]{}

\definecolor{gray-comment}{gray}{0.5}

\theoremstyle{plain}

\makeatletter
\newtheorem*{rep@theorem}{\rep@title}
\newcommand{\newreptheorem}[2]{%
\newenvironment{rep#1}[1]{%
 \def\rep@title{#2 \ref{##1}}%
 \begin{rep@theorem}}%
 {\end{rep@theorem}}}
\makeatother

\newreptheorem{rcor}{Corollary}
\newreptheorem{rthm}{Theorem}

\usepackage{thmtools}
\usepackage{thm-restate}

\DeclareMathOperator{\supp}{supp}

\makeatother

\usepackage{babel}
  \providecommand{\definitionname}{Definition}
  \providecommand{\lemmaname}{Lemma}
  \providecommand{\propositionname}{Proposition}
\providecommand{\theoremname}{Theorem}

\newcounter{note}[section]

\title{Combinatorial Prophet Inequalities}
\author{
Aviad Rubinstein\thanks{(aviad@eecs.berkeley.edu) Electrical Engineering and Computer Sciences, UC Berkeley, Berkeley, CA.
This work was supported by Microsoft Research PhD Fellowship, NSF grant CCF1408635, and Templeton Foundation grant 3966.}
\and Sahil Singla\thanks{(ssingla@cmu.edu) Computer Science Department, Carnegie Mellon University, Pittsburgh, PA. This work was supported by CMU Presidential Fellowship and NSF awards CCF-1319811,
   CCF-1536002, and CCF-1617790}
}
\date{\today}

\begin{document}
\maketitle

\begin{abstract}{
\noindent We  introduce a novel framework of Prophet Inequalities for combinatorial valuation functions.
For a (non-monotone) submodular objective function over an arbitrary matroid feasibility constraint, we  give an $O(1)$-competitive algorithm.
For a monotone subadditive objective function over an arbitrary downward-closed feasibility constraint, we give an $O(\log n \log^2 r)$-competitive algorithm (where $r$ is the cardinality of the largest feasible subset).
%


Inspired by the proof of our subadditive prophet inequality, we also obtain an  $O(\log n \cdot \allowbreak  \log^2 r)$-competitive algorithm for the Secretary Problem with a monotone subadditive objective function subject to an arbitrary downward-closed feasibility constraint. 
Even for the special case of a cardinality feasibility constraint, our algorithm circumvents an $\Omega(\sqrt{n})$ lower bound 
by Bateni, Hajiaghayi, and Zadimoghaddam \cite{BHZ13-submodular-secretary_original} in a restricted query model.

En route to our submodular prophet inequality, we prove a technical result of independent interest:
we show a variant of the Correlation Gap Lemma \cite{CCPV07-correlation_gap, ADSY-OR12} for non-monotone submodular functions.
}\end{abstract}


\ifFULL
\else
\setcounter{page}{0}
\thispagestyle{empty}

\newpage
\fi
\section{Introduction}\label{sec:Intro}

The {\em Prophet Inequality} and {\em Secretary Problem} are classical problems in stopping theory. 
In both problems a decision maker must choose one of $n$ items arriving in an online fashion.
In the Prophet Inequality, each item is drawn independently from a known distribution, but the order of arrival is chosen adversarially. In the Secretary Problem, the decision maker has no prior information about the set of items to arrive (except their cardinality, $n$), 
but the items are guaranteed to arrive in a uniformly random order.

Historically, there are many parallels between the research of those two problems. The classic (single item) variants of both problems were resolved a long time ago: in 1963 Dynkin gave a tight $e$-competitive algorithm for the Secretary Problem \cite{Dynkin63}; a little over a decade later Krengel and Sucheston \cite{KS77-prophet} gave a tight $2$-competitive algorithm for the Prophet Inequality. 
Motivated in part by applications to mechanism design, multiple-choice variants of both problems have been widely studied for the past decade in the online algorithms community. Instead of one item, the decision maker is restricted to selecting a feasible subset of the items.
The seminal papers of \cite{HKP04-first_secretary, Kleinberg05-multiple_choice-secretary} introduced a secretary problem subject to a cardinality constraint (and \cite{Kleinberg05-multiple_choice-secretary} also obtained a $1-O(1/\sqrt{r})$-competitive algorithm). In 2007, Hajiaghayi et al. followed with a prophet inequality subject to cardinality constraint \cite{HKS07-prophet_and_online-MD}. 
In the same year, Babaioff et al. introduced the famous {\em matroid secretary problem} \cite{BIK07-secretary_original}; in 2012 Kleinberg and Weinberg introduced (and solved!) the analogous matroid prophet inequality. For general downward-closed constraints, $O(\log n \log r)$-competitive algorithms were recently obtained for both the problems \cite{Rub16-downward_closed}.

In all the works mentioned in the previous paragraph, the goal is to maximize the sum of selected items' values, i.e. an additive objective is optimized.
For the secretary problem, there has also been significant work on optimizing more general, combinatorial objective functions.
A line of great works \cite{BHZ13-submodular-secretary_original, FNS11-submodular_secretary-old, BUCM12-submodular_secretary-old, FZ15-submodular_secretary} on secretary problem with submodular valuations culminated with a general reduction by Feldman and Zenklusen \cite{FZ15-submodular_secretary} from any submodular function to additive (linear) valuations with only $O(1)$ loss.
Going beyond submodular is an important problem \cite{FI15-supermodular_secretary}, but for subadditive objective functions there is a daunting $\Omega(\sqrt n)$ lower bound on the competitive ratio for restricted value queries \cite{BHZ13-submodular-secretary_original}.

Surprisingly, this line of work on combinatorial secretary problems has seen no parallels in the world of prophet inequalities.
In this work we break the ice by introducing a new framework of combinatorial prophet inequalities.

\paragraph{Combinatorial Prophet Inequalities:}
Our main conceptual contribution is a generalization of the Prophet Inequality setting to combinatorial valuations.
Roughly, on each of $n$ days, the decision maker knows an independent prior distribution over $k$ potential items that could appear.\footnote{Note that some notion of independence assumption is necessary as  even for the single choice problem, if values are arbitrarily correlated then every online algorithm is $\Omega(n)$ competitive~\cite{HillKertz-Journal92}.}
She also has access to a combinatorial (in particular, submodular or monotone subadditive) function $f$ that describes the value of any subset of the $n \cdot k$ items (see Section~\ref{sec:combinatorial} for a formal definition and further discussion).
We obtain the following combinatorial prophet inequalities:

\begin{thm}[Submodular Prophet; informal]\label{thm:submodular-prophet}
There exists an efficient randomized $O(1)$-competitive algorithm  for (non-monotone) submodular prophet over any matroid.
 \end{thm}

\begin{thm}[Monotone Subadditive Prophet; informal]\label{thm:subadditive-prophet}
There exists an $O(\log n \cdot \log^2 r)$-competitive algorithm for monotone subadditive prophet inequality subject to any downward-closed constraints family.
\end{thm}

\paragraph{Subadditive Secretary Problem:}
Building on the techniques of our subadditive prophet inequality, we go back to the secretary world and prove a (computationally inefficient) $O(\log n \cdot \log^2 r)$-competitive algorithm for the subadditive secretary problem subject to any downward-closed feasibility constraint. As noted earlier, this algorithm circumvents the impossibility result of Bateni et al. \cite{BHZ13-submodular-secretary_original} for efficient algorithms%
\footnote{In fact, for general downward-closed constraint, even with additive valuations one should not expect efficient algorithms with membership queries \cite{Rub16-downward_closed}.}.

\begin{thm}[Monotone Subadditive Secretary; informal]\label{thm:subadditive-secretary-informal}
There exists an $O(\log n \cdot \log ^2 r)$-competitive algorithm for monotone subadditive secretaries subject to any downward-closed constraints family.
\end{thm}

\paragraph{Non-monotone correlation gap:}
En route to Theorem \ref{thm:submodular-prophet}, we prove a technical contribution that is of independent interest: a constant {\em correlation gap} for non-monotone submodular functions. 
For a monotone submodular function $f$, \cite{CCPV07-correlation_gap} showed that the expected value of $f$ over any distribution of subsets is at most a constant factor larger than the expectation over subsets drawn from the product distribution with the same marginals.
This bound on the correlation gap has been very useful in the past decade with applications in optimization \cite{CVZ14-CRS}, mechanism design \cite{ADSY-OR12, Yan-SODA11, BCK12-correlation_gap_risk-averse, BH16-correlatin_gap-mechanism_design}, influence in social networks \cite{RSS15-knapsack_adaptive_seeding, BPRS16-locally_adaptive_seeding}, and recommendation systems \cite{KSS13-recommendation_system-correlation_gap}.

It turns out (see Example \ref{ex:naive-corr-gap}) that when $f$ is non-monotone, the correlation gap is unbounded, even for $n=2$!
Instead, we prove a correlation gap for a related function:
\[ f_{\max}(S) \triangleq \max_{T \subseteq S} f(T).\]
(Note that $f_{\max}$ is monotone, but may not be submodular.)

\begin{thm}[Non-monotone correlation gap; informal]\label{thm:corrgap}
For any (non-monotone) submodular function $f$, the function $f_{\max}$ has a correlation gap of $O(1)$.
\end{thm}



\ifFULL
\subsection{Further related work}
\paragraph{Secretary Problem}
In recent years, following the seminal work of Babaioff, Immorlica, and Kleinebrg \cite{BIK07-secretary_original}, there has been extensive work on the {\em matroid secretary problem}, 
where the objective is to maximize the sum of values of secretaries subject to a matroid constraint. 
For general matroids, there has been sequence of improving competitive ratio with the state of the art being $O(\log \log r)$ \cite{Lachish14-secretary, FSZ15-loglogr}.
Obtaining a constant competitive ratio for general matroids remains a central open problem, but constant bounds are known for many special cases (see survey by Dinitz \cite{Dinitz13-survey} and references therein).
Other variants have also been considered, such as a hiring that returns for a second (or potentially $k$-th) interview \cite{Vardi15-returning_secretary}, or secretaries whose order of arrival has exceptionally low entropy \cite{KKN15-limited_randomness}.
Of special interest to us is a recent paper by Feldman and Izsak \cite{FI15-supermodular_secretary} that considers matroid secretary problems with general monotone objective functions, parametrized by the {\em supermodular degree}. For objective $f$ with supermodular degree $\mathcal{D}_f^+$ and general matroid constraint, they obtain a competitive ratio of $O\left({\mathcal{D}_f^+}^3 \log \mathcal{D}_f^+ + {\mathcal{D}_f^+}^2\log r\right)$. Their results are incomparable to our Theorem \ref{thm:subadditive-secretary-informal}, but the motivation is related---obtaining secretary algorithms beyond submodular functions.

\paragraph{Prophet Inequality} 
The connection between multiple-choice prophet inequalities and 
mechanism design was recognized in the seminal paper by Hajiaghayi et al.~\cite{HKS07-prophet_and_online-MD}.
In particular, they proved a prophet inequality for uniform matroids; their bound was later improved by Alaei~\cite{Alaei11-tight_uniform_prophet}.
Chawla et al. further developed the connection between prophet inequalities and mechanism design, and proved, for general matroids, a variant of the prophet inequality where the algorithm may choose the order in which the items are viewed.
The {\em matroid prophet inequality} was first explicitly formulated by Kleinberg and Weinberg \cite{KW12-matroid_prophet}, who also gave a tight $2$-competitive algorithm.
In a different direction, Alaei, Hajiaghayi, and Liaghat~\cite{AHL12-prophet_matching}
considered a variant they call {\em prophet-inequality matching}, which is useful for online ad allocation.
More generally, for intersection of a constant number of matroid, knapsack, and matching constraints, Feldman, Svensson, and Zenklusen \cite{FSZ16-OCRS} gave an $O(1)$-competitive algorithm; this is a corollary of their {\em online contention reslution schemes} (OCRS), which we also use heavily (see Section~\ref{sec:OCRS}).
Azar, Kleinberg, and Weinberg \cite{AKW14-limited-information} considered a limited information variant where the algorithm only has access to samples from each day's distributions.
Esfandiari et al.~\cite{EHLM15-secrotrophet} considered a mixed notion of ``Prophet Secretary'' where the items arrive in a uniformly random order and draw their values from known independent distributions.
Finally, for general downward-closed constraint, \cite{Rub16-downward_closed} gave $O(\log n \log r)$-competitive algorithms for both Prophet Inequality and Secretary Problem; these algorithms are the basis of our algorithms for the respective subadditive problems.

\paragraph{Other notions of online submodular optimization}
Online submodular optimization has been studied in  contexts beyond secretary. In online submodular welfare maximization, there are $m$ items, $n$ people, and each person has a monotone submodular value function. Given the value functions, the items are revealed one-by-one and the problem is to immediately and irrevocably allocate it to a person, while trying to maximize the  sum of all the value functions (welfare). The greedy strategy is already half competitive. Kapralov et al.~\cite{KPV-SODA13} showed that for adversarial arrival greedy is the best possible in general (competitive ratio of $1/2$), but under a ``large capacities'' assumption, a primal-dual algorithm can obtain $1-1/e$-competitive ratio \cite{D0KMY13-large_capacities}. For random arrival Korula et al.~\cite{KMZ-STOC15} showed that greedy can beat half; obtaining $1-1/e$ in this settings remains open. 

Buchbinder et al.~\cite{BFS-SODA15} considered the problem of (monotone) submodular maximization with {\em preemption}, when the items are revealed in an adversarial order. Since sublinear competitive ratio is not possible in general with adversarial order, they consider a relaxed model where we are allowed to drop items (preemption) and give constant-competitive algorithms. Submodular maximization has also been studied in the \emph{streaming setting}, where we have  space constraints but are again  allowed to drop items~\cite{BMKK-KDD14,CK-IPCO15,CGQ-ICALP15}.

The ``learning community'' has looked into \emph{experts} and \emph{bandits} settings for submodular optimization. In these settings,   different submodular functions arrive one-by-one and the algorithm, which is trying to minimize/ maximize its value, has to select a set before seeing the function. The function is then revealed and the algorithm gets the  value for the selected set. The goal is to perform as close as possible to the best fixed set in hindsight. Since submodular minimization can be reduced to convex function minimization using Lov\'asz extension, sublinear regrets are possible~\cite{HazanKale12}. For submodular maximization, the usual benchmark is a $1-1/e$ multiplicative loss and an additive regret~\cite{GolovinStreeter-NIPS08,GolovinKrause-COLT10,GKS-ArXiv14}.

Interestingly, to the best of our knowledge none of those problems have been studied for subadditive functions.

\else
\subsection{Related work}
Due to space constraints, we defer to the full version important discussion on connections to related works on secretary problem~\cite{BIK07-secretary_original, Lachish14-secretary, FSZ15-loglogr, Dinitz13-survey, Vardi15-returning_secretary, KKN15-limited_randomness, FI15-supermodular_secretary},
prophet inequality~\cite{HKS07-prophet_and_online-MD, Alaei11-tight_uniform_prophet, KW12-matroid_prophet, AHL12-prophet_matching, FSZ16-OCRS, AKW14-limited-information, EHLM15-secrotrophet, Rub16-downward_closed}, and other related online models~\cite{KPV-SODA13, D0KMY13-large_capacities, KMZ-STOC15, BFS-SODA15, BMKK-KDD14,CK-IPCO15,CGQ-ICALP15, HazanKale12, GolovinStreeter-NIPS08,GolovinKrause-COLT10,GKS-ArXiv14}.
\fi

\subsection{Organization}
We begin by defining our model for combinatorial prophet inequalities in Section~\ref{sec:combinatorial}; 
in Section~\ref{sec:Prelim} we develop some necessary notation and recall known results;
in Section~\ref{sec:pf-cor-gap} we formalize and prove our correlation gap for non-monotone submodular functions; and
in Section~\ref{sec:ProphetMatroid} we prove the submodular prophet inequality.

The subadditive prophet and secretary algorithms share the following high level approach:
use a lemma of Dobzinski \cite{Dobzinski07-subadditive-vs-XOS} to reduce subadditive to XOS objective functions;
then solve the XOS case using the respective (prophet and secretary) algorithms of \cite{Rub16-downward_closed} for additive objective function and general downward closed feasibility constraint.
It turns out that the secretary case is much simpler since we can use the additive downward-closed algorithm of \cite{Rub16-downward_closed} as black-box; we present it in Section~\ref{sec:Secretary}.
For the prophet inequality, we have to make changes to the already highly non-trivial algorithm of \cite{Rub16-downward_closed} for the additive, downward-closed case; we defer this proof 
\ifFULL
to Section~\ref{sec:Prophet}.
\else
to the full version.
\fi

\section{Combinatorial but Independent Functions}\label{sec:combinatorial}
In this section, we define and motivate what we mean by ``submodular (or subadditive) valuations over independent items''.
Recall that in the classic prophet inequality, a gambler is asked to choose one of $n$ independent non-negative random payoffs.
In {\em multiple choice prophet inequalities}~\cite{HKS07-prophet_and_online-MD}, the gambler chooses multiple payoffs, subject to some known-in-advance feasibility constraint, and receives their sum.
Here, we are interested in the case where the gambler's utility is not additive over the outcomes of the random draws: 
For example, instead of monetary payoffs, at each time period the gambler can choose to receive a random item, say a car, and the utility from owning multiple cars  diminishes quickly. 
As another example, the gambler receives monetary payoffs, but his marginal utility for the one-millionth dollar is much smaller than for the first.
Formally, we define:
\begin{defn}
[{${\cal C}$ Valuations over Independent Items}]\label{def:C-independent}
Let ${\cal C}$ be a class of valuation functions (in particular, we are interested in ${\cal C} \in \{\text{submodular, monotone subadditive}\}$).
Consider:
\begin{itemize}
\item $n$ sets $U_1, \dots ,U_n$ and distributions ${\cal D}_1, \dots ,{\cal D}_n$, where ${\cal D}_i$ returns a single item from $U_i$. 
\item A function $f:\{0,1\}^{U} \rightarrow \mathbb{R}_+$, where $U \triangleq  \bigcup_{i=1}^n U_i$ and $f \in {\cal C}$.
\item For $\X \in \prod_i U_i$ and subset $S \subseteq [n]$, let 
$v_{\X}(S) \triangleq f\big(\{X_i: i\in S\}\big)$.
\end{itemize}

Let ${\cal D}$ be a distribution over valuation functions $v\left(\cdot\right):2^{n}\rightarrow\mathbb{R}_+$.
We say that ${\cal D}$ is ``{\em ${\cal C}$  
over independent items}'' if it can  
be written as the distribution that first samples $\X \sim \bigtimes_i {\cal D }_i$, and then outputs valuation function $v_{\X}(\cdot)$.

\end{defn}

\subsubsection*{Related notions in the literature}

Combinatorial functions over independent items have been considered before. 
Agrawal et al. \cite{ADSY-OR12}, for example, consider a different, incomparable framework defined via the generalization of submodular and monotone functions to non-binary, ordered domains 
(e.g.  $f: \bigtimes U_i \rightarrow \mathbb{R}_+$ is submodular if  $f(\max\{\X, \Y\}) + f(\min\{\X, \Y\}) \leq f(\X) + f(\Y)$).
In our definition, per contra, there is no natural way to define a full order over the set of items potentially available on each time period. 
For example, if we select an item on Day~$1$, an item $X_2$ on Day~$2$ may have a larger marginal contribution than item $Y_2$, 
but a lower contribution if we did not select any item on Day~$1$. 

Another relevant definition has been considered in probability theory \cite{Sch99-concentration_results} and more recently in mechanism design \cite{RW15-subadditive}.
The latter paper considers auctioning $n$ items to a buyer that has  a random, ``independent'' monotone subadditive valuation over the items.
Now the seller knows which items she is selling them, but different types of buyers may perceive each item differently. 
This is captured via an {\em attribute} of an item, which describes how each buyer values a bundle containing this item.
Formally,

\begin{defn}
[{Monotone Subadditive Valuations over Independent Items \cite{Sch99-concentration_results, RW15-subadditive}}]\label{def:subadditive-independent-RW}
We say that a distribution ${\cal D}$ over valuation functions $v\left(\cdot\right):2^{n}\rightarrow\mathbb{R}$
is {\em subadditive over independent items} if:
\begin{enumerate}
\item All $v\left(\cdot\right)$ in the support of ${\cal D}$ exhibit \emph{no externalities}.

Formally, let $\Omega_{S}=\bigtimes_{i\in S}\Omega_{i}$, where each
$\Omega_{i}$ is a compact subset of a normed space. There exists
a distribution ${\cal D}^{\X}$ over $\Omega_{[n]}$ and functions
$V_{S}:\Omega_{S}\rightarrow\mathbb{R}$ such that ${\cal D}$ is
the distribution that first samples $\X\leftarrow{\cal D}^{\X}$
and outputs the valuation function $v\left(\cdot\right)$ with $v\left(S\right)=V_{S}\left(\langle X_{i}\rangle_{i\in S}\right)$
for all $S$.

\item All $v\left(\cdot\right)$ in the support of ${\cal D}$ are monotone and subadditive.
\item The private information is \emph{independent across items}. That
is, the ${\cal D}^{\X}$ guaranteed in Property 1 is a product
distribution.
\end{enumerate}
\end{defn}

{ \noindent For monotone valuations, Definition~\ref{def:C-independent} is stronger than Definition~\ref{def:subadditive-independent-RW} as it assumes that the valuation function is defined over every subset of $U = \bigcup U_i$, rather than just the support of ${\cal D}$. However, it turns out that for monotone subadditive functions Definitions \ref{def:C-independent} and \ref{def:subadditive-independent-RW} are equivalent.}
\begin{observation}
A distribution ${\cal D}$ is subadditive over independent items according to Definition \ref{def:C-independent} if and only if 
it is subadditive over independent items according to Definition \ref{def:subadditive-independent-RW}.
\end{observation}
\ifFULL
\begin{proof}[Proof sketch]
It's easy to see that Definition \ref{def:C-independent} implies Definition \ref{def:subadditive-independent-RW}: 
We can simply identify between the set of attributes $\Omega_i$ on day $i$ and the set of potential items $U_i$, 
and let $V_{S}\left(\langle X_{i}\rangle_{i\in S}\right) \triangleq f\big(\{X_i: i\in S\}\big)$.
Observe that the  desiderata of Definition \ref{def:subadditive-independent-RW} are satisfied.

In the other direction, we again identify between each $U_i$ and $\Omega_i$.
For feasible set $S$ which consists of only one item in $U_i$ for each $i\in R$, we can let $\X_S$ be the corresponding vector in $\prod \Omega_i$, and define 
$f\big(S\big) \triangleq V_{\X_S}(R)$.
Definition \ref{def:C-independent} requires that we define $f(\cdot)$ over any subset of $U \triangleq \bigcup U_i$.
We do this by taking the maximum of $f(\cdot)$ over all feasible subsets. Namely, for set $T \subseteq U$, let $R_T \subseteq [n]$ denote again the set of $i$'s such that $|T \cap U_i|  \geq 1$. We set:
\[
	f(T) \triangleq \max_{\substack{S\subseteq T \;\;\; \text{s.t.}\\ \forall i \;\;\; |S \cap U_i| \leq 1}} V_{\X_S}(R_T).
\]
Now $f(\cdot)$ is monotone  subadditive because it is  maximum of monotone  subadditive functions.
\end{proof}
\else
(See full version for details.)
\fi

\section{Preliminaries}\label{sec:Prelim}

\subsection{Submodular and Matroid Preliminaries}\label{sec:submodular-prelims}
A set function $f: \{0,1\}^n \rightarrow \mathbb{R}_+$ is a submodular function if for all $S,T \subseteq [n]$ it satisfies $f(S\cup T) + f(S \cap T) \leq f(S) + f(T)$. For any $e\in [n]$ and  $S \subseteq [n]$, let $f_S(e)$  denote $f(S\cup e) - f(S)$.
We recall some  notation for to extend submodular functions from the discrete hypercube $\{0,1\}^n$ to relaxations whose domain is the continuous hypercube $[0,1]^n$.

For any vector $\x \in [0,1]^n$, let $S \sim \x$  denote a random set $S$ that contains each element $i \in [n]$ independently w.p. $x_i$. Moreover, let $1_S$ denote a vector of length $n$ containing $1$ for $i \in S$ and $0$ for $i \not\in S$.

\begin{defn} We define important continuous extensions of any set function $f$.

\noindent \quad \emph{Multilinear extension $F$}:
\begin{align*}
F(\x) &\triangleq \E_{S\sim \x}[f(S)].
\intertext{\quad \emph{Concave closure $f^+$}:}
 f^+(\x) &\triangleq \max_{\alpha}\Big\{ \sum_{S\subseteq [n]} \alpha_S f(S) \mid \sum_S \alpha_S =1 \text{ and }\sum_S \alpha_S 1_S = \x \Big\}.
  \intertext{\quad \emph{Continuous relaxation $f^*$}:}
 f^*(\x) &\triangleq    \min_{S \subseteq [n]} \bigg\{ f(S) + \sum_{i \in [n] \setminus S} f_S(e)\cdot x_i  \bigg\} .
\end{align*}
\end{defn}

\subsubsection{Some Useful Results}
\begin{lem}[Correlation gap \cite{CCPV07-correlation_gap}] \label{lem:corrgapmonot} For any monotone submodular function and $\x \in [0,1]^n$, 
\[ F(\x) \leq f^+(\x) \leq f^*(\x) \leq \left(1 - \frac1e \right)^{-1} F(\x). 
\]
\end{lem}

\begin{lem}[Lemma $2.2$ of~\cite{BFNS-SODA14}]\label{lem:BFNS}
Consider any submodular function $f$ and any set $A \subseteq [n]$. Let $S$ be a random subset of $A$ that contains each element of $S$ w.p. at most $p$ (not necessarily independently), then
\[ \E_{S} [f(S)] \geq (1-p) \cdot f(\emptyset).
\]
\end{lem}

\ifFULL 
\begin{lem}[Lemma $2.3$ of~\cite{FMV-SICOMP11}] \label{lem:FMV-doubleLem}
For any non-negative submodular function $f$ and any sets $A,B \subseteq [n]$,
\[ 
\E_{ \substack{S \sim 1_A/2 \\T \sim 1_B/2}} [f(S \cup T)] \geq \frac14 \left( f(\emptyset) + f(A) + f(B) + f(A \cup B)\right).
\]
\end{lem}

\begin{lem}[Theorem $2.1$ of~\cite{FMV-SICOMP11}] \label{lem:FMV-double}
For any non-negative submodular function $f$, any set $A \subseteq [n]$ and $p \in [0,1]$,
\[ 
F({1}_A \cdot p ) \geq p(1-p) \cdot \max_{T \subseteq A} f(T).
\]
\end{lem}

We can now prove the following useful variant of the previous two lemmata
(see also an alternative self-contained proof in Appendix~\ref{sec:missingProofs}).

\else
We also need the following lemma, which is inspired by similar lemmata of Feige, Mirrokni, and Vondrak~\cite{FMV-SICOMP11}.
 \fi
\begin{lem}\label{lem:low-high}
Consider any non-negative submodular function $f$  and $0 \leq L \leq H \leq 1$. Let $S^*$ be a set that maximizes $f$, and let $\x \in [0,1]^n$ be such that for all $i \in [n]$, $L \leq x_i \leq H$.
Then, \[F(\x) \geq  L(1-H)\cdot f(S^*).\]
\end{lem}
\ifFULL
\begin{proof}
Imagine a process in which in which we construct the random set $T \sim \x$,
i.e. set $T$ containing each element $e$ independently w.p. $x_e$, in two steps.
In the first step we construct set $T'$ by selecting every element independently w.p. exactly $L$. In the second step we construct set $\tilde{T}$  containing each element $e$  independently  w.p. $(x_e - L)/(1-L)$.
It's easy to verify that the union of the two sets $T' \cup \tilde{T}$ contains each element $e$ independently with  probability exactly $x_e$.

From Lemma~\ref{lem:FMV-double}, we know that at the end of first step, the generated set has
expected value  $\E[f(T')] \geq L (1-L) f(S^*)$. Now, we argue that the second step does not ``hurt'' the value by a lot. We note that in the second step each element is added w.p. at most $(H - L)/(1-L)$  because $x_e \leq H$. 
Let $g(S) := f(S \cup T')$  be a non-negative submodular function. We apply Lemma~\ref{lem:BFNS} on $g$ to get $\E[g(\tilde{T})] \geq (1-H)/(1-L)\cdot g(\emptyset) $, which implies $\E[f(T' \cup \tilde{T})] \geq (1-H)/(1-L)\cdot \E[f(T')]$.

Together,  we get  $F(\x) = \E[T' \cup \tilde{T}] \geq  L (1-L) (1-H)/(1-L) f(S^*) =
L (1-H) f(S^*)$.
\end{proof}
\else
\fi

\subsubsection{Online Contention Resolution Schemes}\label{sec:OCRS}

Given a point $\x$ in the matroid polytope $P$ of matroid $\mathcal{M}$,  many submodular maximization  applications   like to select each element $i$ independently with probability $x_i$ and claim that the selected set $S$ has expected value $F(\x)$~\cite{CVZ14-CRS}. The difficulty is that $S$ need not be feasible in $\mathcal{M}$, and we can only select  $T \subseteq S$ that is feasible. Chekuri et al.~\cite{CVZ14-CRS} introduced the notion of \emph{contention resolution schemes} (CRS) that describes how, given a random $S$, one can find a feasible $T\subseteq S$ such that the expected value  $f(T)$ will be close to $F(\x)$.

Recently, Feldman, Svensson, and Zenklusen gave {\em online contention resolution schemes} (OCRS). Informally, it says that the decision of whether to select element $i \in S$  into $T$ can be made online, even before knowing the entire set $S$~\cite{FSZ16-OCRS}. 
In particular, we will need their definition of {\em greedy} OCRS; we define it below and state the results from \cite{FSZ16-OCRS} that we use in our $O(1)$-submodular prophet inequality result over matroids.

\begin{defn}[Greedy OCRS] Let $\x$ belong to a matroid polytope $P$ and  $S \sim \x$. A greedy OCRS defines a downward-closed family $\cal{F}_{\x}$ of  feasible  sets in the matroid. All elements reveal one-by-one if they belong to $S$, and  when element $i \in [n]$ reveals, the greedy OCRS selects it if, together with the already selected elements, the obtained set is in $\cal{F}_{\x}$.
\end{defn}

\ignore{
\begin{defn}[$(b,c)$-selectable] For $0 \leq b,c \leq 1$, a greedy OCRS is $(b,c)$-selectable if for any $\x \in b\cdot P $, given that some element $i \in S$, where $S\sim \x$, it selects  $i $ with probability at least $c$. 
%
\end{defn}

\begin{lem}[Theorem $1.8$ of~\cite{FSZ16-OCRS}]\label{lem:OcrsExist}
There exists a $(\frac12,\frac12 )$-selectable deterministic greedy OCRS for matroid polytopes. 
\end{lem}

\begin{lem}[Theorem $1.10$ of~\cite{FSZ16-OCRS}]\label{lem:OcrsUseful}
Given a non-negative submodular function $f$ and a $(b,c)$-selectable greedy OCRS for a polytope $P$, applying OCRS to an input $x \in b\cdot P$ results in a random set $T$ satisfying $\E_T[F(1_T/2)] \geq (c/4)\cdot F(x)$.
\end{lem}}

\begin{lem}[Theorems 1.8 and 1.10 of~\cite{FSZ16-OCRS}]\label{lem:OCRS}
Given a non-negative submodular function $f$, a matroid $\mathcal{M}$, and a vector $\x$ in the convex hull of independent sets in $\mathcal{M}$, there exists a deterministic greedy OCRS that outputs a set $T$ satisfying $\E_T[F(1_T/2)] \geq (1/16)\cdot F(x)$.
\end{lem}

\subsection{Subadditive and Downward-Closed Preliminaries} 
 A set function $f: \{0,1\}^n \rightarrow \mathbb{R}_+$ is {\em subadditive} if for all $S,T \subseteq [n]$ it satisfies $f(S\cup T) \leq f(S) + f(T)$. It's monotone if $f(S) \leq f(T)$ for any $S \subseteq T$.
 A  set function $f: \{0,1\}^n \rightarrow \mathbb{R}_+$ is an {\em XOS}  (alternately, {\em fractionally subadditive}) function if there exist linear functions $L_i: \{0,1\}^n \rightarrow \mathbb{R}_+$ such that $f(S) = \max_i \{L_i(S) \}$. (See Feige~\cite{Feige-SICOMP09} for illustrative examples.)

\begin{lem}
[{\cite[Lemma 3]{Dobzinski07-subadditive-vs-XOS}}]\label{lem:dobzinski}
Any subadditive function $v:2^{n}\rightarrow\mathbb{R}_{+}$ can be
$\left(\frac{\log\left|S\right|}{2e}\right)$-approximated by an XOS
function $\widehat{v}:2^{n}\rightarrow\mathbb{R}_{+}$; i.e. for every
$S\subseteq\left[n\right]$,
\[
\widehat{v}\left(S\right)\leq v\left(S\right)\leq\left(\frac{\log\left|S\right|}{2e}\right)\widehat{v}\left(S\right).
\]
 Furthermore, the XOS function has the form $\widehat{v}\left(S\right)=\max_{T\subseteq\left[n\right]}p_{T}\cdot\left|T\cap S\right|$
for an appropriate choice of $p_{T}$'s.\end{lem}

\begin{thm} [{\cite{Rub16-downward_closed}}] \label{thm:additive_secretary}
When items take values in $\left\{ 0,1\right\}$,
there  are $O\left(\log n\right)$-competitive algorithms for (additive) {\sc Downward-Closed Secretary} and {\sc Downward-Closed Prophet}.
\end{thm}

\section{Correlation Gap for non-monotone submodular functions}\label{sec:pf-cor-gap}

For monotone submodular functions, \cite{CCPV07-correlation_gap} proved that 
\begin{gather} \label{eq:monotone-gap}
F(\x) \geq (1-1/e) f^+(\x). 
\end{gather}
This result was later rediscovered by \cite{ADSY-OR12}, who called the ratio between $f^+(\x)$ and $F(\x)$ {\em correlation gap}.
It's useful in many applications since it says that up to a constant factor, picking items independently is as good as the best correlated distribution with the same element marginals.

What is the correct generalization of \eqref{eq:monotone-gap} to non-monotone submodular functions?
It is tempting to conjecture that $F(\x) \geq c\cdot f^+(\x)$ for some constant $c>0$.
However, the following example shows that even for a function as simple as the directed cut function on a two-vertex graph, this gap may be unbounded.
\begin{example} \label{ex:naive-corr-gap}
Let $f$ be the directed cut function on the two-vertex graph $u \rightarrow v$; i.e. $f(\emptyset) = 0$, $f(\{u\}) = 1$, $f(\{v\}) = 0$, and $f(\{u,v\}) = 0$.
Let $\x = (\epsilon, 1-\epsilon)$. Then, 
\[F(\x) = \epsilon^2 \ll \epsilon = f^+(\x).\]
\end{example}

It turns out that the right way to generalize \eqref{eq:monotone-gap} to non-monotone submodular functions is to first make them monotone:

\begin{defn} [$f_{\max}$]
\[ f_{\max}(S) \triangleq \max_{T \subseteq S} f(T).\]
\end{defn}
For non-monotone submodular $f$, we have that $f_{\max}$ is monotone, but it may no longer be submodular, as shown by the following example:
\begin{example} [$f_{\max}$ is not submodular]
Let $f$ be the directed cut function on the four-vertex graph $u \rightarrow v \rightarrow w \rightarrow x$. 
In particular,  $f(\{v\}) = 1$, $f(\{u,v\}) = 1$,  $f(\{v,w\}) = 1$, and $f(\{u,w\}) = 2$.
\[f_{\max}(\{u,v\}) - f_{\max}(\{v\})  = 1 - 1 < 2 -1 = f_{\max}(\{u,v,w\}) - f_{\max}(\{v,w\}).\]
\end{example}

Finally, we are ready to define correlation gap for non-monotone functions:

\begin{defn} [Correlation gap] \label{def:correlation-gap}
The {\em correlation gap} of any set function $f$ is
\[ \max_{\x \in [0,1]^n}  \max_{\alpha \geq 0  } \left\{ \frac{ f^+(\x) }{ F_{\max} (\x) } \Bigl\vert \sum_S \alpha_S = 1 \text{ and } \sum_S \alpha_S 1_S = \x \right\}, 
\]
where $F_{\max}$ is the multilinear extension of $f_{\max}$.
\end{defn}

Notice that for monotone $f$, we have that $f_{\max} \equiv f$, so  Definition~\ref{def:correlation-gap} generalizes the correlation gap for monotone submodular functions.
Furthermore, one could  replace $f^+$  with $f^+_{\max}$ in Definition~\ref{def:correlation-gap}; observe that the resulting definition is equivalent.

\begin{thm}[Non-monotone correlation gap]\label{thm:correlation-gap_formal}
For any (non-monotone non-negative) submodular function $f$, the  correlation gap is at most $200$.
\end{thm}

While the constant can be improved slightly, we have not tried to optimize it, focusing instead on clarity of exposition.
Our proof goes through a third relaxation, $f^*_{1/2}$. 

\begin{defn} [$f^*_{1/2}$] \label{defn:gstar} 

For any set function $f$ and any $\x\in [0,1]^n$, 
\begin{align*} 
 f^*_{1/2}(\x) \triangleq  \min_{S \subseteq [n]} \bigg\{ \E_{T \sim 1_S/2} \big[f(T) + \sum_{i \in [n] \setminus S} f_T(e)\cdot x_i \big] \bigg\}. 
\end{align*}
\end{defn}

Below, we will prove (Lemma \ref{lem:gstaratleast}) that $f^+(\x) \leq 4 \cdot f^*_{1/2}(\x)$.
We then show (Lemma \ref{lem:gstaratmost}) that $f^*_{1/2}(\x) \leq 50 \cdot F(\x/2)$, which implies
\begin{gather}\label{eq:x/2-correlation_gap}
f^+(\x) \leq 4 \cdot f^*_{1/2}(\x) \leq 200 \cdot F(\x/2) .
\end{gather}
 
Finally, to finish the proof of Theorem~\ref{thm:correlation-gap_formal}, it suffices to show that $F(\x/2) \leq F_{\max}(\x)$. This is easy to see since drawing $T$ according to $\x/2$ is equivalent to drawing $S$ according to $\x$, and then throwing out each element from $S$ independently with probability $1/2$. For $F_{\max}(\x)$, on the other hand, we draw the same set $S$ and then take the optimal subset.

\subsection{Proof that $f^+(\x) \leq 4 \cdot f^*_{1/2}(\x)$}

\begin{lem}\label{lem:gstaratleast}
For any $\x\in [0,1]^n$ and non-negative submodular function $f: \{0,1\}^n \rightarrow \mathbb{R}_+$,
\[ f^+(\x) \leq 4 \cdot f^*_{1/2}(\x).
\]\end{lem}

We first prove the following auxiliary claim:
\begin{claim} \label{cla:aux}
For any sets $S,T \subseteq [n]$,
\[ \E_{T_{1/2} \sim 1_T/2} [f((S \setminus T) \cup T_{1/2}] \geq \frac14 f(S). \]
\end{claim}
\begin{proof}
Define a new auxiliary function $h(U) \triangleq f((S \setminus T) \cup U)$. Observe that $h$ continues to be non-negative and submodular.
We now have,
\begin{align*}
\E_{T_{1/2} \sim 1_T/2} [f((S \setminus T) \cup T_{1/2}] 
& = \E_{T_{1/2} \sim 1_T/2} [h(T_{1/2})] \\
& \geq \frac14 h(T) && \text{(Lemma~\ref{lem:low-high} for $L=H=1/2$)}\\
& = \frac14 f(S) && \text{(Definition of $h$)}.
\end{align*}

\ignore{
Define also $A \triangleq T \cap S$ and $B \triangleq T \setminus S$.
We can now rewrite, 
\begin{gather*}
\E_{T_{1/2} \sim 1_T/2} [f((S \setminus T) \cup T_{1/2}] = \E_{T_{1/2} \sim 1_T/2} [h(T_{1/2})] 
 = \E_{ \substack{S \sim 1_A/2 \\T \sim 1_B/2}} [h(S \cup T)] .
\end{gather*}
Applying Lemma \ref{lem:FMV-doubleLem}, we get: 
\begin{align*}
\E_{T_{1/2} \sim 1_T/2} [f((S \setminus T) \cup T_{1/2}] & \geq \frac14 \left( h(\emptyset) + h(A) + h(B) + h(A \cup B)\right)  \\
& \geq \frac14 h(A) && \text{($h$ is non-negative)}\\
& = \frac14 f((S \setminus T) \cup A) && \text{(Definition of $h$)}\\
& = \frac14 f(S) && \text{(Definition of $A$).}
\end{align*}
}

\end{proof}

\begin{proof}[Proof of Lemma \ref{lem:gstaratleast}]
Fix $\x$, and let $S^* = S^*(\x)$ denote the optimal set that satisfies $f^*_{1/2}(\x) =  \E_{T \sim 1_{S^*}/2} \big[f(T) + \sum_{i \in [n] \setminus S^*} f_T(e) x_i \big] $. 
Let $\{\alpha_S\}$ be the optimal distribution that satisfies $f^+(\x) = \sum_S \alpha_S f(S)$. 
Then, $\frac 14  f^+(\x) = \frac 14 \sum_S \alpha_S f(S) $
\begin{align*}
&\leq  \sum_S \alpha_S \cdot \E_{T \sim 1_{S^*}/2} [ f((S\setminus S^*)\cup T) ] && \text{(Claim \ref{cla:aux})} \\
&=  \sum_S \alpha_S \cdot \E_{T \sim 1_{S^*}/2} \left[ f(T) + f_T(S \setminus S^*) \right] \\
&\leq  \sum_S \alpha_S \cdot \E_{{T \sim 1_{S^*}/2}} \left[ f(T) + \sum_{i \in S\setminus S^*} f_T(e) \right]  && \text{(submodularity)}\\
&=   \E_{{T \sim 1_{S^*}/2}} \left[f(T) \sum_S \alpha_S   +  \sum_S \alpha_S \sum_{i \in S\setminus S^*} f_T(e)  \right]   \\
&= \E_{{T \sim 1_{S^*}/2}} \left[  f(T) +   \sum_{i \in [n]\setminus S^*} f_T(e) x_i \right]  = f^*_{1/2} (\x)	&&  \text{(using $\sum_S \alpha_S 1_S = \x$)}.
\end{align*}
\end{proof}

\subsection{Proof that $f^*_{1/2}(\x) \leq 50 \cdot F(\x /2)$}

The proof of the following lemma is similar to Lemma $5$ in~\cite{CCPV07-correlation_gap}. \ifFULL \else (See full version). \fi

\begin{lem}\label{lem:gstaratmost}.
\[   f^*_{1/2}(\x) \leq 50 \cdot F\left( \x / 2\right).
\]
\end{lem}
\ifFULL
\begin{proof}
Consider an exponential clock running for each element $i \in [n]$ at rate $x_i$. Whenever the clock triggers, we update set $S$ to $S \cup \{i\}$. For $t \in [0,1]$, let $S(t)$ denote the set of elements in $S$ by time $t$. Thus, each element belongs to $S(1)$ w.p. $1 - \exp(-x_i)$, which is between $x_i (1- \frac1e)$ and $x_i$. Let $V(t) \triangleq \E_{T \sim 1_{S(t)}/2 } [f(T)]$, i.e. expected value of set that picks each element in $S(t)$ independently w.p. $\frac12$. 
Our goal is to show that:
\begin{gather} \label{eq:goal}  
 f^*_{1/2}(\x) \leq \left(\frac{2}{1-e^{-1/2}} \right) \cdot \E [V(1)] \leq \left(\frac{2}{1-e^{-1/2}} \right) \left(\frac{4(e-1)}{e-2} \right) \cdot  F\left( \x / 2\right).
\end{gather}

We begin with the second inequality of \eqref{eq:goal}. 
Consider the auxiliary submodular function $g(S) \triangleq \E_{T \sim 1 - \exp(-\x)} [f(S \cap T)]$, and let $G$ denote its multilinear extension.
Let $S^*$ be a maximizer of $g$, and observe that 
\begin{gather*} 
V(1) = G(1_{[n]}/2) \leq g(S^*).
\end{gather*}
Observe further that, with slight abuse of notation, $F(\x / 2) = G\left(\frac{\x/2}{1 - \exp(-\x)} \right)$; 
this is well-defined 
since for any $x_i \in [0,1]$, we have 
\begin{gather*}
\frac12 \leq  \frac{x_i/2}{1 - \exp(-x_i)} \leq  \frac{e}{2(e-1)} < 1.
\end{gather*}
Moreover, since  $\frac{x_i/2}{1 - \exp(-x_i)}$ is bounded,  Lemma~\ref{lem:low-high} gives 
\[F(\x / 2) \geq \frac12 \cdot \left(1- \frac{e}{2(e-1)}\right) \cdot g(S^*) = \frac{e-2}{4(e-1)} g(S^*)  =\Omega(1) \cdot g(S^*).\]

We now turn to the first inequality of \eqref{eq:goal}.
Consider an infinitesimal interval interval $(t, t+dt]$. For any $i \notin S(t)$ the exponential clock triggers with probability $x_i \; dt$, so it contributes to $V(t+dt)$ with probability $x_i/2 \; dt$. The probability that two clocks trigger in the same infinitesimal is negligible ($O(dt^2)$). Therefore,
\begin{align*} \E[ V(t+dt) - V(t)]  & = \E_{S(t)} \E_{T \sim 1_{S(t)}/2 }  \left[ \sum_{j \in [n]\setminus S} \frac{x_i}{2} f_T(j) \; dt  \right]  -O(dt^2)\\
& \geq  \frac{1}{2} \Big(f^*_{1/2}(\x) - \underbrace{\E_{S(t)} \E_{T \sim 1_{S(t)}/2 } \E[f(T)]}_{\E[ V(t) ]}\Big)\; dt -O(dt^2).
\end{align*}
Dividing both sides by $dt$ and taking the limit as $dt \rightarrow 0$, we get:
\begin{gather*}
\frac{d}{dt} \E[ V(t) ]  \geq  \frac{1}{2} \Big(f^*_{1/2}(\x) - \E[ V(t) ]\Big).
\end{gather*}

To solve the differential inequality, let $\phi(t) =  \E[V(t) ] $ and $\psi(t) = \exp(\frac{t}{2})\, \phi(t)$. We get $\frac{d\phi}{dt} \geq \frac12 (f^*_{1/2}(\x) - \phi(t))$ and $\frac{d\psi}{dt} = \exp(\frac{t}{2}) (\frac{d\phi}{dt} + \frac{\phi(t)}{2}) \geq \exp(\frac{t}{2})\,\frac{f^*_{1/2}(\x)}{2}$. Since $\psi(0) = \phi(0) = 0$, integration over $t$ gives
\[ \E[V(t)] = \phi(t) = \exp(-t/2) \,\psi(x) \geq \frac{f^*_{1/2}(\x)}{2} (1 - \exp(-t/2) ).
\]
In particular, plugging in $t=1$ completes the proof of the first inequality in \eqref{eq:goal}.
\end{proof}
\else
\fi

\section{Submodular Prophets over Matroids}\label{sec:ProphetMatroid}
\begin{defn} [{\sc Submodular Matroid Prophet}]
The offline inputs to the problem are: 
\begin{itemize}
\item $n$ sets $U_1, \dots ,U_n$; we denote their union $U \triangleq \bigcup_{i=1}^n U_i$;
\item a (not necessarily monotone) non-negative submodular function $f:\{0,1\}^{U} \rightarrow \mathbb{R}_+$;
\item $n$ distributions $\mathcal{D}_i$ over subset $U_i$; and
\item a matroid $\cal{M}$ over $[n]$
\end{itemize}
On the $i$-th time period, the algorithm observes an element $X_i \in U_i$ drawn according to $\mathcal{D}_i$, independently from outcomes and actions on past and future days. The algorithm must decide (immediately and irrevocably) whether to add $i$ and $X_i$ to sets $W$ and $X_W$, respectively, subject to  $W$ remaining independent in $\cal{M}$. 
The objective is to maximize $f(X_W)$.
\end{defn}  
\begin{thm}\label{thm:submod-prophet-matroid}
There is a randomized algorithm with a competitive ratio of $O(1)$ for any
{\sc Submodular Matroid Prophet}
 \end{thm}
 
\paragraph{Proof overview} The main ingredients in the proof of Theorem \ref{thm:submod-prophet-matroid} are known {\em online contention resolution schemes} (OCRS) due to Feldman, Svensson, and Zenklusen~\cite{FSZ16-OCRS}, and our new bound on the {\em correlation gap} for non-monotone submodular functions (Theorem~\ref{thm:correlation-gap_formal}).

Let $\x \in [0,1]^{U}$ denote the vector of probabilities that each element realizes (i.e. $x_{(i,j)} = \mathcal{D}_i(j)$). A naive proof plan proceeds as follows:
Select elements online using the OCRS (w.r.t $\x$); obtain a constant factor approximation to $F(\x)$; use a ``correlation gap'' to show a constant factor approximation of $f^+(\x)$; finally, observe that $f^+(\x)$ is an upper bound on $OPT$.

There are two problems with that plan: First, the OCRS of Feldman et al. applies when elements realize independently.
The realization of different elements for the same day is obviously correlated (exactly one element realizes), so we cannot directly apply their OCRS.
The second problem is that for non-monotone submodular function, it is in general not true that $F(\x)$ approximates $f^+(\x)$ (see Example \ref{ex:naive-corr-gap}). 

The solution to both obstacles is working with $\x/2$ instead of $\x$. 
In Section~\ref{sec:pf-cor-gap} we showed that $F(\x/2)$ is a constant factor approximation of $f^+(\x)$ (Ineq. \eqref{eq:x/2-correlation_gap}). 
Then, 
in Subsection \ref{sec:use-OCRS}, we give an algorithm that approximates the selection of the greedy OCRS on $\x/2$.
Our plan is then to show:
\begin{align*}
ALG & = \Omega(\E_{S \sim OCRS(\x/2)}[f(S)]) && \text{(Subsection \ref{sec:use-OCRS})} \\
	& = \Omega(F(\x/2)) && \text{(Lemma \ref{lem:OCRS})} \\
	& = \Omega (f^+(\x)) && \text{(Ineq. \eqref{eq:x/2-correlation_gap})} \\
	& = \Omega (OPT).
\end{align*}

\subsection{Applying the OCRS to our setting}\label{sec:use-OCRS}

In this subsection we show an algorithm that obtains, in expectation, $1/2$ of the expected value of the OCRS with probabilities $\x/2$.


Our algorithm uses the greedy OCRS  as a black box. On each day, the algorithm (sequentially) feeds the OCRS a subset of the elements $U_i$ that can potentially arrive on that day. The subset on each day is chosen at random; it is correlated with the element that actually arrives on that day, and independent from the subsets chosen on other days.
The guarantee is that the distribution over sequences fed into the OCRS is identical to the distribution induced by $\x/2$.

\subsubsection*{Reduction}

For each $i$, let $U_i$ denote the set of elements that can arrive on day $i$, and fix some (arbitrary) order over $U_i$.  
For a subset $S_i \subseteq U_i$, let $P_{\x/2}^i(S_i)$ denote the probability that the set $S_i$ is exactly the outcome of sampling from $U_i$ according to $\x/2$.
When element $(i,j)$ arrives on day $i$, the algorithm feeds into the OCRS a random set $T_i$ drawn from the following distribution. With probability $\frac{P_{\x/2}^i(\{(i,j)\})}{x_{i,j}}$, the algorithm feeds just element $(i,j)$, i.e.  $T_i = \{(i,j)\}$; notice that this guarantees $ \Pr \left[T_i = \{(i,j)\} \right]= P_{\x/2}^i(\{(i,j)\})$.
Otherwise, the algorithm lets $T_i$ be a random subset of $U_i$, drawn according to $\x/2$, conditioned on $|T_i| \neq 1$. This guarantees that the probability mass on subsets of size $\neq 1$ is also allocated according to $\x/2$.

Now, if the algorithm fed the singleton $\{(i,j)\}$ and the OCRS selected it, then the algorithm also takes $\{(i,j)\}$; otherwise the algorithm does not take $\{(i,j)\}$. (In particular, if $|T_i| \neq 1$, the algorithm ignores the decisions of the OCRS.)

\subsubsection*{Analyzing the reduction}

Observe that on each day the distribution over $T_i$'s is identical to the distribution $P_{\x/2}^i(\cdot)$. Since the $T_i$'s are also independent, it means that the distribution of inputs to the OCRS is indeed distributed according to $\x/2$.

Conditioning on $(i,j)$ is being fed (i.e., with probability $x_{i,j}/2$), $P_{\x/2}^i(\cdot)$  assigns  at least $1/2$ probability to the event where no other element is also being fed (this is precisely the reason we divide $\x$ by $2$):  
\[
\Pr[T_i = \{(i,j)\} \mid  T_i \ni (i,j)] \geq 1/2.
\]


Since the OCRS is greedy, for any history on days $1,\dots,i-1$, if it selects $(i,j)$ when observing set $T_i \ni (i,j)$, it would also select $(i,j)$ when observing only this element on day $i$. 
Furthermore, since the OCRS is only allowed to select one element on day $i$, conditioning on the OCRS selecting $(i,j)$, the future days ($i+1,\dots,n$) proceed independently of whether the algorithm also selected $(i,j)$.
Therefore, conditioning on the greedy OCRS selecting any set $S_{\textrm{OCRS}}$, 
the algorithm selects a subset $T_{\textrm{ALG}} \subseteq S_{\textrm{OCRS}}$ where each element appears with probability at least $1/2$.

Finally, to argue that the algorithm obtains at least $1/2$ of the expected value of the set selected by the OCRS, fix the set $S_{\textrm{OCRS}}$ selected by the OCRS, and consider the submodular function
$g(\bar{T}) \triangleq f(S_{\textrm{OCRS}} \setminus \bar{T})$. Setting $\bar{T} \triangleq T_{\textrm{ALG}} \setminus S_{\textrm{OCRS}}$, we have that $f(T_{\textrm{ALG}}) = g(\bar{T})$.
Thus by Lemma \ref{lem:BFNS}, 
\begin{gather*}
\E[f(T_{\textrm{ALG}})] \geq \frac{1}{2} \E[g(\emptyset)] =  \frac{1}{2}\E[f(S_{\textrm{OCRS}})]. \qed
\end{gather*}


\section{Subadditive Secretary over Downward-Closed Constraints}\label{sec:Secretary}

\begin{defn} [{\sc Monotone Subadditive Downward-Closed Secretary}] 
Consider $n$ items, a monotone subadditive valuation function from subsets of items to $\mathbb{R}_+$, and an arbitrary downward-closed set system
${\cal F}$ over the items; both $f$ and ${\cal F}$
are adversarially chosen. 
The algorithm receives as input $n$ (but
not ${\cal F}$ or $f$). The items arrive in
a uniformly random order. Initialize $W$ as the empty set. When item
$i$ arrives, the algorithm observes all feasible subsets of items that have already arrived, and their valuation in $f$.
The algorithm then decides (immediately and irrevocably) whether to add $i$ to the set $W$, subject to the constraint that $W$ remains a feasible set in ${\cal F}$. The goal is to maximize $f(W)$. 
\end{defn}

\begin{thm}
\label{thm:subadditive-secretary}There is a deterministic algorithm
for {\sc Monotone Subadditive Downward-Closed Secretary} that achieves
a competitive ratio of $O\left(\log n\cdot\log^{2}r\right)$.
\end{thm}

\begin{proof}
Let $T^{\star}$ be the set chosen by the offline algorithm ($OPT = f(T^{\star})$). By Lemma \ref{lem:dobzinski} there exists a $p_{T^{\star}}$ such that for every $S\subseteq T^{\star}$:
  \begin{align}
  \label{eq:f(S)>pt|S|} f(S) &\geq p_{T^{\star}}\left|S\cap T^{\star}\right|; \\
\label{eq:OPT<pt|T|logr} OPT &= f(T^{\star}) = O\left( p_{T^{\star}} \left|T^{\star}\right| \log\left|T^{\star}\right|\right) = O\left(p_{T^{\star}} \left|T^{\star}\right|\right) \log r.
\end{align}

Assume that we know $p_{T^{\star}}$ (discussed later). 
We define a new feasibility constraint ${\cal F}'$ as follows: a set $T \subseteq [n]$ is feasible in ${\cal F}'$ iff it is feasible in ${\cal F}$ and for every subset $S \subseteq T$, we have $f(S) \geq p_{T^{\star}}\left|S\right|$. Notice that because we also force the condition on all subsets of $T$, ${\cal F}'$ is downward-closed and it does not depend on the order of arrival.

We run the algorithm for $\{0,1\}$-valued (additive) {\sc Downward-Closed Secretary} (as guaranteed by Theorem \ref{thm:additive_secretary}) with feasibility constraint ${\cal F}'$ where all values are $1$.
By \eqref{eq:f(S)>pt|S|}, $T^{\star}$ is feasible in ${\cal F}'$, and by \eqref{eq:OPT<pt|T|logr} $ p_{T^{\star}} \left|T^{\star}\right| = \Omega\left(\frac{OPT}{\log r}\right)$.
Therefore, the additive $\{0,1\}$-values algorithm returns a set $T^{ALG}$ of size $\left|T^{ALG}\right| = \Omega\left(\frac{OPT}{p_{T^{\star}}  \log n \log r}\right)$. Furthermore, $T^{ALG}$ is also feasible in ${\cal F}'$, i.e. 
\begin{gather}\label{eq:t^ALG}
f(T^{ALG}) \geq p_{T^{\star}} \left|T^{ALG}\right| = \Omega\left(\frac{OPT}{\log n \log r}\right).
\end{gather}

\subsubsection*{Guessing $p_{T^{\star}}$}
Finally, we don't actually know $p_{T^{\star}}$, but we can guess it correctly, up to a constant factor, with probability $1/\log r$. 
\ifFULL
We run the classic secretary algorithm over the first $n/2$ items, where we use the value of the singleton $f(\{i\})$ as ``the value of item $i$'': Observe the first $n/4$ items and select none; then take the next item whose value is larger than every item observed so far. 
With constant probability this algorithm selects the item with the largest value, which we denote by $M$. 

Also, with constant probability the algorithm sees the item with the
largest value too early and does not select it. Assume that this is
the case. Since we obtained expected value of $\Omega\left(M\right)$
on the first $n/2$ items we can, without loss of generality, ignore
values less than $M/r$.
In particular, we know that $p_{T^{\star}} \in [M/r, M]$.
Pick $\alpha \in \{M/r,M/(2r),\dots,M/2,M\}$ uniformly at random, and use it instead of $p_{T^{\star}}$ to define ${\cal F}'$.  
With probability $1/\log r$, $p_{T^{\star}} \in [\alpha, 2\alpha]$, in which case the algorithm returns a set $T^{ALG}$ satisfying \eqref{eq:t^ALG}.
\else
See full version for details.
\fi
\end{proof}

\ifFULL
\else
\bibliographystyle{alpha}
\bibliography{prophet2}
\end{document}
\fi

\section{Subadditive Prophet}\label{sec:Prophet}

\begin{defn} [{\sc Monotone Subadditive Downward-Closed Prophet}]
The offline inputs to the problem are: 
\begin{itemize}
\item $n$ sets $U_1, \dots ,U_n$; we denote their union $U \triangleq \bigcup_{i=1}^n U_i$;
\item a monotone non-negative subadditive function $f:\{0,1\}^{U} \rightarrow \mathbb{R}_+$;
\item $n$ distributions $\mathcal{D}_i$ over subset $U_i$; and
\item a feasibility constraint $\cal{F}$ over $[n]$.
\end{itemize}
On the $i$-th time period, the algorithm observes an element $X_i \in U_i$ drawn according to $\mathcal{D}_i$, independently from outcomes and actions on past and future days. The algorithm must decide (immediately and irrevocably) whether to add $i$ and $X_i$ to sets $W$ and $X_W$, respectively, subject to the constraint that $W$ remains feasible in $\cal{F}$. 
The objective is to maximize $f(X_W)$.
\end{defn}  

Let $r$ denote the maximum cardinality of a feasible set $S\in{\cal F}$. 
\begin{thm}
\label{thm:subadditive-prophet}There is a deterministic algorithm
for {\sc Monotone Subadditive Downward-Closed Prophet} that achieves
a competitive ratio of $O\left(\log n\cdot\log^{2}r\right)$.
\end{thm}
The proof of Theorem \ref{thm:subadditive-prophet} consists of three
steps: in Subsection \ref{sub:Subadditive-to-XOS} we reduce  monotone
subadditive valuations over independent items to monotone XOS subadditive
valuations over independent items, with a loss of $O\left(\log r\right)$,
using a lemma of Dobzinski \cite{Dobzinski07-subadditive-vs-XOS}.
Then in Subsection \ref{sub:XOS-to-XOS} we use a standard reduction
from general XOS valuations to XOS with $\left\{ 0,1\right\} $ marginal
contributions, losing another factor of $O\left(\log r\right)$. Finally,
in Subsection \ref{sub:XOS-with-01} we use techniques from \cite{Rub16-downward_closed}
to give an $O\left(\log n\right)$-competitive algorithm for monotone
XOS with $\left\{ 0,1\right\} $ marginal contributions.

\subsection{Subadditive to XOS \label{sub:Subadditive-to-XOS}}

\begin{defn} [{\sc Monotone XOS Downward-Closed Prophet}]
For any set $M$ and  items [n], the offline inputs to the problem are: 
\begin{itemize}
\item $n$ sets $U_1, \dots ,U_n$ of {\em valuations vectors} in $\mathbb{R}_+^M$; we denote their union $U \triangleq \bigcup_{i=1}^n U_i$;
\item a monotone XOS function $\widehat{f}:\{0,1\}^{U} \rightarrow \mathbb{R}_+$ 
\[\widehat{f}(S) \triangleq \max_{m \in M} \sum_{\myvec{u} \in S}u_m \text{ for $S \in \{0,1\}^{U}$};\]
\item $n$ distributions $\mathcal{D}_i$ over subset $U_i$; and
\item a feasibility constraint $\cal{F}$ over $[n]$, which is a collection of subsets of $[n]$.
\end{itemize}
On the $i$-th time period, the algorithm observes a valuations vector $X_i \in U_i$ drawn according to $\mathcal{D}_i$, independently from outcomes and actions on past and future days. The algorithm must decide (immediately and irrevocably) whether to add $i$ and $X_i$ to sets $W$ and $X_W$, respectively, subject to the constraint that $W$ remains feasible in $\cal{F}$. 
The objective is to maximize $\widehat{f}(X_W)$.
\end{defn}  

Below (Proposition \ref{prop:XOS-prophet-general}) we give an $O\left(\log n\cdot\log r\right)$-competitive algorithm for {\sc Monotone XOS Downward-Closed Prophet}.
By Dobzinski's lemma (Lemma \ref{lem:dobzinski}), this implies an $O\left(\log n\cdot\log^2 r\right)$-competitive algorithm for {\sc Monotone Subadditive Downward-Closed Prophet}.

\begin{prop}
\label{prop:XOS-prophet-general}There is a deterministic algorithm
for {\sc Monotone XOS Downward-Closed Prophet} that achieves a competitive
ratio of $O\left(\log n\cdot\log r\right)$.\end{prop}

\subsection{XOS to XOS with $\left\{ 0,1\right\} $ coefficients\label{sub:XOS-to-XOS}}
Below (Proposition \ref{prop:XOS-prophet-01}), we give an $O\left(\log n\right)$-competitive algorithm for {\sc Monotone XOS Downward-Closed Prophet} in the special case where all the vectors $v \in U$ are in $\{0,1\}^M$.
First, let us show why this would imply Proposition \ref{prop:XOS-prophet-general}.

\begin{proof}
[Proof of Proposition \ref{prop:XOS-prophet-general} from Proposition \ref{prop:XOS-prophet-01}]
We recover separately the contributions from ``tail'' events (a single
item taking an exceptionally high value) and the ``core'' contribution
that is spread over many items. Run the better of the following two
algorithms:

\paragraph*{Tail }

Let $OPT$ denote the expected offline optimum value. Whenever we
see a feasible item whose valuations vector $X_{i}$ has value at least
$2OPT$, we select it. For item $i$, let $p_{i}=\Pr\left[X_{i}\geq2OPT\right]$.
We have 
\[
OPT\geq2OPT\cdot\Pr\left[\exists i\colon X_{i}\geq2OPT\right]=2OPT\cdot\left(1-\prod\left(1-p_{i}\right)\right).
\]
Dividing by $OPT$ and rearranging, we get
\[
1/2\leq\prod\left(1-p_{i}\right)\leq e^{-\sum p_{i}},
\]
and thus 
\[
\sum p_{i}\leq\ln2.
\]

Therefore the probability that we want to take an item but can't is
at most $\ln2$, so this algorithm achieves at least a $\left(1-\ln2\right)$-fraction
of the expected contribution from values greater than $2OPT$. 

\paragraph*{Core }

Observe that we can safely ignore values less than $OPT/2r$, as those
can contribute a total of at most $OPT/2$. Partition all remaining
values into $2+\log r$ intervals $\left[OPT/2r,OPT/r\right],\allowbreak\dots,\mbox{\allowbreak}\left[OPT,2OPT\right]$.
The expected contribution  from the values in each interval is $\Omega\left(1/\log r\right)$-fraction
of the expected offline optimum without values greater than $2OPT$.
Pick the interval with the largest expected contribution, round down
all the values in this interval, and run the algorithm guaranteed
by Proposition \ref{prop:XOS-prophet-01}. This achieves an $\Omega\left(\frac{1}{\log n\cdot\log r}\right)$-fraction
of the expected contribution from values less than or equal to $2OPT$.
\end{proof}

\subsection{XOS with $\left\{ 0,1\right\} $ coefficients\label{sub:XOS-with-01}}
\begin{prop}
\label{prop:XOS-prophet-01}When the $X_{i}$'s take values in $\left\{ 0,1\right\} ^{M}$,
there is a deterministic algorithm for {\sc Monotone XOS Downward-Closed Prophet}
that achieves a competitive ratio of $O\left(\log n\right)$.
\end{prop}

\subsubsection{A dynamic potential function}\label{sec:dynpotent}

At each iteration, the algorithm maintains a target value $\tau$
and a target probability $\pi$, where $\pi$ is the probability (over future
realizations) that the current restricted prophet beats $\tau$. We
say that an outcome (i.e. a pair of item and valuations vector) is {\em good}
if selecting it does not decrease the probability of beating the target
value by a factor greater than $n^{2}$, and {\em bad} otherwise.
Notice that all the bad items together contribute at most a $\left(1/n\right)$-fraction
of the probability of beating $\tau$. A key ingredient is that $\tau$
is updated dynamically. If the probability of observing a good outcome
is too low (less than $1/4$), we deduct $1$ from $\tau$. We show
(Lemma \ref{lem:main}) that this increases $\pi$ by a factor of
at least $2$. Since $\pi$ decreases by at most an $n^{2}$ factor when we select
an item, and increases by a factor of $2$ whenever we deduct $1$
from $\tau$: we balance $2\log n$ deductions for every item the
algorithm selects, and this gives the $O\left(\log n\right)$ competitive
ratio. 

So far our algorithm is roughly as follows: set a target value $\tau$;
whenever the probability $\pi$ of reaching the target $\tau$ drops
below $1/4$, decrease $\tau$; if $\pi>1/4$, sit and wait for a
good outcome - one will arrive with probability at least $1/4$ (we actually do this with $\Pr[A]$ instead of $\pi$, where $A$ is a closely related event). There
is one more subtlety: what should the algorithm do if no good outcomes
arrive? In other words, what if the probability of observing a good
outcome is neither very low nor very close to $1$, say $1/2$ or
even $1-\frac{1}{\log n}$? On one hand, we can't decrease $\tau$
again, because we are no longer guaranteed a significant increase
in $\pi$; on the other hand, after, say $\Theta\left(\log^{2}n\right)$
iterations, we still have a high probability of having an iteration
where none of the good outcomes arrive. (If no good outcomes are coming,
we don't want the algorithm to wait forever...) Fortunately, there
is a simple solution: the algorithm waits for the last item with a
good outcome in its support; if, against the odds, no good outcomes
have yet been observed, the algorithm ``hallucinates'' that this last
item has a good valuations vector, and selects it. In expectation, at most
a constant fraction of the items we select are ``hallucinated'', so
the competitive ratio is still $O\left(\log n\right)$.

\subsubsection{Notation}

We let $OPT$ denote the expected (offline) optimum. $W$ is the set
of items selected so far ($W$ for ``Wins''), and $\ell_{W}\triangleq\max\left\{ i\in W\right\} $
is the index of the last selected item.

Let ${\cal F}$ denote the family of all feasible subsets of $[n]$. For any $T\subseteq\left[n\right]$,
let ${\cal F}_{T}$ denote the family of feasible sets whose intersection
with $\left\{ 1,\dots,\max\left( T\right) \right\} $ is exactly
$T$. 

Let $X_{i}=\left(X_{i}^{m}\right)_{m\in M}\in\left\{ 0,1\right\} ^{M}$
denote the random vector drawn for the $i$-th item. We use $z_{i}^ {}$
to refer to the observed realization of $X_{i}^ {}$. Our algorithm will maintain a subset $M' \subseteq M$. We let 
\[
V_{M'}\left({\cal F},X_{\left[n\right]}\right)\triangleq\max_{S\in{\cal F}}\max_{m\in M'}\sum_{i\in S}(X_{i})_{m}
\]
 denote the value of optimum offline solution (note that this is also
a random variable). 

Let $\tau=\tau\left(W\right)$ be the current target value, and $\pi=\pi\left(\tau,W\right)$
denotes the current target probability:
\[
\pi\left(\tau,W\right)\triangleq\Pr\left[V_{M'}\left({\cal F}_{W},X_{\left[n\right]}^ {}\right)>\tau\mid X_{\left[\ell_{W}\right]}^ {}=z_{\left[\ell_{W}\right]}^ {}\right].
\]
For each $y_{j}\in\supp\left(X_{j}^ {}\right)$, we define $\pi^{j,y_{j}}=\pi^{j,y_{j}}\left(\tau,W\right)$
to be the probability of reaching $\tau$, given that:
\begin{itemize}
\item $z_{j}=y_{j}$,
\item $j$ is the next item we select, and 
\item item $j$ actually contributes $1$ to the offline optimum. 
\end{itemize}
Formally, 
\[
\pi^{j,y_{j}}\left(\tau,W\right)\triangleq\Pr\left[V_{M'\cap y_{j}}\left({\cal F}_{W+j},X_{\left[n\right]}\right)>\tau\mid X_{\left[\ell_{W}\right]\cup\left[j\right]}=\left(z_{\left[\ell_{W}\right]},y_{j}\right)\right],
\]
where we slightly abuse notation and also use $y_{j}$ to denote the
set of $m\in M$ such that $y_{j}^{m}=1$. 

We say that a future outcome $\left(j,y_{j}\right)$ is {\em good}
if $\pi^{j,y_{j}}\geq n^{-2}\cdot\pi$ and $j$ is feasible (and otherwise
it is {\em bad}), and let $G=\left\{ \mbox{good\,}\left(j,y_{j}\right)\right\} $
denote the set of good future outcomes. Finally, 
\[A \triangleq A\left(\pi,\tau,W\right), \]
\text{ is the event that at least one of the good outcomes occur.} 

\subsubsection{Updated proof plan and the algorithm}
The  idea is to always maintain a threshold $\tau$ such that probability of one of the good outcomes to occur is large, i.e. $\Pr[A]$ is at least a constant $\frac14$. The way we do this is by showing  in Claim~\ref{claim:Api} that at any time during the execution of the algorithm, conditioned on what all has happened till now, the probability $\pi$ that the offline algorithm achieves the threshold $\tau$ gives a lower bound on $\Pr[A]$. 
Hence, whenever $\Pr[A]$ goes below $\frac14$, we  decrease the threshold $\tau$, which increases $\pi$ due to Lemma~\ref{lem:main} and, indirectly, increases $\Pr[A]$  by Claim~\ref{claim:Api}.

Initialize $\tau\leftarrow OPT/2$, $M'\leftarrow M$, and $W\leftarrow\emptyset$. 
Lemma~\ref{lem:concentration} uses  a concentration bound due to Ledoux to show that in the beginning $\tau = OPT/2$ satisfies $\pi > \frac14$.

After each update to $W$, decrease $\tau$ until $\Pr\left[A\right]\geq1/4$,
or until $\left|W\right|>\tau$. When $\Pr\left[A\right]\geq1/4$,
reveal the values of items until observing a good outcome. When we
observe a good outcome $z_{j}$, add $j$ to $W$ and restrict $M'$
to its intersection with $z_{j}$. Since we restrict $M$ to $M'$, this gives us that at any time \[V_{M'}\left({\cal F},X_{W}\right) = |W|. \]
If we reach the last item with
good outcomes in its support, and none of the good outcomes realize, add
this last item to $G$ and subtract $1$ from $\tau$ 
 (without modifying $M'$). See also pseudocode in Algorithm \ref{alg:prophet}.

\begin{algorithm}
\protect\caption{\label{alg:prophet}Prophet}

\begin{enumerate}
\item $\tau\leftarrow\frac{OPT}{2}$; $M'\leftarrow M$; $W\leftarrow\emptyset$
\item while $\tau>\left|W\right|$:

\begin{enumerate}
\item $\pi\leftarrow\Pr\left[V_{M'}\left({\cal F}_{W},X_{\left[n\right]}\right)>\tau\mid X_{\left[\ell_{W}\right]}=z_{\left[\ell_{W}\right]}\right]$

{\color{gray-comment} \# $\pi$ is the probability that, given the
history, the offline optimum can still beat $\tau$.}

\item $G\leftarrow\left\{ \left(j,y_{j}\right):j>\ell_{W}\,\mbox{AND\,}\pi^{j,y_{j}}\geq n^{-2}\cdot\pi\right\} \cap\left(\bigcup_{S\in{\cal F}_{W}}S\right)$

{\color{gray-comment} \# $G$ is the set of good and feasible outcomes.}

\item if $\Pr\left[A\right]\geq1/4$

{\color{gray-comment} \# A good outcome is likely occur.}
\begin{enumerate}
\item $j^{*}\leftarrow\min\left\{ j\in G\colon\left(j,z_{j}\right)\in G\right\} $

{\color{gray-comment} \# Wait for a good and feasible outcome.}

\item if $j^{*}=\infty$

{\color{gray-comment} \# No good outcomes.}
\begin{enumerate}
\item $j^{*}\leftarrow\max G$

{\color{gray-comment} \# Select the last potentially good item.}

\item $\tau\leftarrow\tau-1$

{\color{gray-comment} \# Adjust the target value to account for select
an item with value $0$}

\end{enumerate}
\item else

{\color{gray-comment} \# $j^{*}$ is actually a good item.}
\begin{enumerate}
\item $M'\leftarrow M'\cap z_{j}$
\end{enumerate}
\item $W\leftarrow W\cup\left\{ j^{*}\right\} $
\end{enumerate}
\item else

\begin{enumerate}
\item $\tau\leftarrow\tau-1$\label{enu:deduct-from-tau}

{\color{gray-comment} \# decrease target value $\tau$ until $\Pr\left[A\right]\geq1/4.$}\end{enumerate}
\end{enumerate}
\end{enumerate}
\end{algorithm}

We first claim that $\pi$ gives us a lower bound on $\Pr[A]$ because most of the mass in $\pi$ comes from good outcomes.
\begin{claim}\label{claim:Api}
At any point during the run of the algorithm, 
\[ \Pr[A\left(\pi,\tau,W\right)] \geq \left(1- \frac1n \right) \pi(W,\tau).
\]
\end{claim}
\begin{proof}
For each $\left(j,y_{j}\right)\notin G$, we have,
by definition of $G$, 
\[
\pi^{j,y_{j}}\left(W,\tau\right)<n^{-2}\cdot\pi\left(W,\tau\right).
\]
Summing over all $\left(j,y_{j}\right)\notin G$,

\begin{flalign*}
\sum_{j}\sum_{y_{j}:\left(j,y_{j}\right)\notin G}\Pr\left[y_{j}\right]\cdot\pi^{j,y_{j}}\left(W,\tau\right) & \leq\sum_{j}\sum_{y_{j}:\left(j,y_{j}\right)\notin G}\Pr\left[y_{j}\right]\cdot\left(n^{-2}\cdot\pi\left(W,\tau\right)\right)\\
 & \leq\sum_{j}n^{-2}\cdot\pi\left(W,\tau\right)\\
 & \leq n^{-1}\cdot\pi\left(W,\tau\right).
\end{flalign*}
Thus, most of $\pi$ comes from good $\left(j,y_{j}\right)$'s:
\begin{equation}
\Pr[A] = \sum_{j}\sum_{y_{j}:\left(j,y_{j}\right)\in G}\Pr\left[y_{j}\right]\cdot\pi^{j,y_{j}}\left(W,\tau\right)\geq\left(1-1/n\right)\pi\left(W,\tau\right).\label{eq:most-pi-from-good}
\end{equation}
\end{proof}

\subsubsection{Concentration for the beginning} \label{sec:concsubaddproph}
\begin{thm}
{\cite[Theorem 2.4]{Ledoux1997}}\label{thm:ledoux} There exists
some constant $K>0$ such that the following holds. Let $Y_{i}$'s
be independent (but not necessarily identical) random variables in
some space $S$; let ${\cal C}$ be a countable class of measurable
functions $f\colon S\rightarrow\left[0,1\right]$; and let $Z=\sup_{f\in{\cal C}}\sum_{i=1}^{n}f\left(Y_{i}\right)$.
Then, 
\[
\Pr\left[Z\geq\E\left[Z\right]+t\right]\le\exp\left(-\frac{t}{K}\cdot\log\left(1+\frac{t}{\E\left[Z\right]}\right)\right).
\]

\end{thm}
To make the connection to our setting, let $Y_{i}$ be the vector
in $\left[0,1\right]^{{\cal F}\times M}$ whose $\left(S,m\right)$-th
coordinate is $X_{i}^{m}$ if $i\in S$, and $0$ otherwise. Let $f_{S,m}\left(Y_{i}\right)\triangleq\left[Y_{i}\right]_{S,m}$,
so $\sum_{i=1}^{n}f_{S,m}\left(Y_{i}\right)$ is simply the value
of the set $S$ under the $m$-th summation in the XOS representation
of the valuation function. Let ${\cal C}\triangleq\left\{ f_{S}\right\} _{S\in{\cal F}}$.
The above concentration inequality can now be written as
\begin{equation}
\Pr\left[V\left({\cal F},X_{\left[n\right]}\right)\geq OPT+t\right]\le\exp\left(-\frac{t}{K}\cdot\log\left(1+\frac{t}{OPT}\right)\right).\label{eq:V-concentrates}
\end{equation}

\begin{lem}
\label{lem:concentration}Assume $OPT\geq\Omega\left(\log n\right)$.
Then,
\[
\Pr\left[V\left({\cal F},X_{\left[n\right]}\right)\geq\frac{OPT}{2}\right]>1/4.
\]
\end{lem}
\begin{proof}
We have, 
\begin{align}
OPT & =\int_{-OPT}^{\infty}\Pr\left[V\left({\cal F},X_{\left[n\right]}\right)\geq OPT+t\right]dt,\label{eq:OPT-integral}
\end{align}
which can be decomposed as to integrals over $\left[-OPT,-OPT/2\right]$,
$\left[-OPT/2,OPT\right]$, and $\left[OPT,\infty\right]$. 

The first two integrals can be easily bounded as
\[
\int_{-OPT}^{-OPT/2}\Pr\left[V\left({\cal F},X_{\left[n\right]}\right)\geq OPT+t\right]dt\leq\int_{-OPT}^{-OPT/2}1\cdot dt\leq\frac{OPT}{2}
\]
and
\begin{eqnarray*}
\int_{-OPT/2}^{OPT}\Pr\left[V\left({\cal F},X_{\left[n\right]}\right)\geq OPT+t\right]dt & \leq & \int_{-OPT/2}^{OPT}\Pr\left[V\left({\cal F},X_{\left[n\right]}\right)\geq OPT/2\right]dt\\
 & \leq & \frac{3OPT}{2}\cdot\Pr\left[V\left({\cal F},X_{\left[n\right]}\right)>\frac{OPT}{2}\right].
\end{eqnarray*}

For the third integral we use the concentration bound (\ref{eq:V-concentrates}):
\begin{eqnarray*}
\int_{OPT}^{\infty}\Pr\left[V\left({\cal F},X_{\left[n\right]}\right)\geq OPT+t\right]dt & \leq & \int_{OPT}^{\infty}\exp\left(-\frac{t}{K}\cdot\log\left(1+\frac{t}{OPT}\right)\right)dt\\
 & \leq & \int_{OPT}^{\infty}\exp\left(-\frac{t}{K}\right)dt\\
 & = & \left[Ke^{-t/K}\right]_{OPT}^{\infty}=K\cdot e^{-OPT/K},
\end{eqnarray*}
which is negligible since $OPT=\omega\left(1\right)$. 

Plugging into (\ref{eq:OPT-integral}), we have:
\[
OPT\leq\frac{OPT}{2}+\frac{3OPT}{2}\cdot\Pr\left[V\left({\cal F},X_{\left[n\right]}\right)>\frac{OPT}{2}\right]+o\left(1\right),
\]
and after rearranging we get
\[
\Pr\left[V\left({\cal F},X_{\left[n\right]}\right)>\frac{OPT}{2}\right]\geq1/3-o\left(1\right).
\]

\end{proof}

\subsubsection{Main lemma}

\begin{lem}
\label{lem:main}At any point during the run of the algorithm, if
$\Pr\left[A\right]\leq1/4$, then subtracting $1$ from $\tau$ doubles
$\pi$; i.e. 
\[
\pi\left(W,\tau-1\right)\geq2\pi\left(W,\tau\right).
\]
\end{lem}

\begin{proof}[Proof of Lemma~\ref{lem:main}]
Consider the event that the optimum solution (conditioned on the items
$W$ we already selected and the realizations $z_{\left[\ell_{W}\right]}$
we have already seen) reaches $\tau$. We can write it as a union
of disjoint events, depending on the next item $j>\ell_{W}$ that
is part of the optimum solution, and its possible realizations $y_{j}$:
\[
\pi\left(W,\tau\right)=\sum_{j}\sum_{y_{j}}\Pr\left[y_{j}\right]\cdot\underbrace{\Pr\left[V_{M'\cap y_{j}}\left({\cal F}_{W\cup\left\{ j\right\} },X_{\left[n\right]}\right)>\tau\mid X_{\left[\ell_{W}\right]\cup\left[j\right]}=\left(z_{\left[\ell_{W}\right]},y_{j}\right)\right]}_{\pi^{j,y_{j}}\left(W,\tau\right)}.
\]
We break the RHS into the sum over $\left(j,y_{j}\right)$'s that
are good and the sum over those that are bad. Now, Claim~\ref{claim:Api} gives 

\begin{equation}
\sum_{j}\sum_{y_{j}:\left(j,y_{j}\right)\in G}\Pr\left[y_{j}\right]\cdot\pi^{j,y_{j}}\left(W,\tau\right)\geq\left(1-1/n\right)\pi\left(W,\tau\right).\label{eq:most-pi-from-good}
\end{equation}

Since $y_{j}\in\left\{ 0,1\right\} ^{M}$, each item can contribute
at most $1$ to the offline optimum. Therefore:

\[
\underbrace{\Pr\left[V_{M'\cap y_{j}}\left({\cal F}_{W\cup\left\{ j\right\} },X_{\left[n\right]}\right)>\tau\mid X_{\left[\ell_{W}\right]\cup\left[j\right]}=\left(z_{\left[\ell_{W}\right]},y_{j}\right)\right]}_{\pi^{j,y_{j}}=\pi^{j,y_{j}}\left(\tau,W\right)}\leq\pi^{j,0}\left(W,\tau-1\right)
\]
Plugging into (\ref{eq:most-pi-from-good}), we have
\begin{eqnarray}
\left(1-1/n\right)\pi\left(W,\tau\right) & \leq & \sum_{j}\sum_{y_{j}:\left(j,y_{j}\right)\in G}\Pr\left[y_{j}\right]\cdot\pi^{j,0}\left(W,\tau-1\right)\nonumber \\
 & \leq & \sum_{j}\Pr\left[\left(j,y_{j}\right)\in G\right]\cdot\pi^{j,0}\left(W,\tau-1\right)\nonumber \\
 & \leq & \left(\sum_{j}\Pr\left[\left(j,y_{j}\right)\in G\right]\right)\cdot\pi\left(W,\tau-1\right),\label{eq:sum-of-Pr=00005Bj,y=00005D}
\end{eqnarray}
where the second inequality follows because $\pi^{j,0}\left(W,\tau-1\right)$
doesn't depend on $y_{j}$, and the third because conditioning on
the $j$-th item being $0$ can only decrease the probability of reaching
$\tau-1$. 

Recall that $A$ is the union of all the events $\left(j,y_{j}\right)\in G$.
Therefore,
\[
\Pr\left[A\right]\geq\sum_{j}\Pr\left[\left(j,y_{j}\right)\in G\right]\left(1-\Pr\left[A\right]\right)
\]
Plugging in $\Pr\left[A\right]<1/4$, we get that $\sum_{j}\Pr\left[\left(j,y_{j}\right)\in G\right]<1/3$.
Plugging into (\ref{eq:sum-of-Pr=00005Bj,y=00005D}) and rearranging,
we get
\[
\pi\left(W,\tau-1\right)\geq\frac{3n}{n-1}\pi\left(W,\tau\right).
\]
 
\end{proof}

\subsubsection{Putting it all together}
\begin{lem}
\label{lem:potential}At any point during the run of the algorithm,
\[
\tau\geq\frac{OPT}{2}-\left(2\log n+1\right)\cdot\left|W\right|-2
\]
\end{lem}
\begin{proof}
We prove by induction that at any point during the run of the algorithm,
\begin{equation}
\log\pi\geq-2-\left(2\log n+1\right)\cdot\left|W\right|+\left(\frac{OPT}{2}-\tau\right).\label{eq:induction}
\end{equation}
After initialization, $\log\pi\geq-2$ by Lemma \ref{lem:concentration}.
By definition of $G$, whenever we add an item to $W$, we decrease
$\log\pi$ by at most $2\log n$ - hence the $2\log n\cdot\left|W\right|$
term. Notice that when the algorithm ``hallucinates'' a $1$, we also
decrease $\tau$ by $1$ to correct for the hallucination - at any
point during the run of the algorithm, this has happened at most $\left|W\right|$
times. Recall that we may also decrease $\tau$ in the last line of
Algorithm \ref{alg:prophet} (in order to increase $\pi$); whenever
we do this, $\tau$ decreases by $1$, but $\pi$ doubles (by Lemma
\ref{lem:main}), so $\log\pi$ increases by $1$, and Inequality
(\ref{eq:induction}) is preserved. 

Finally, since $\pi$ is a probability, we always maintain $\log\pi\leq0$.
\end{proof}
We are now ready to complete the proof of Theorem \ref{thm:subadditive-prophet}.
\begin{proof}
[Proof of Proposition~\ref{prop:XOS-prophet-01}]The algorithm always terminates
after at most $O\left(OPT\right)$ decreases to the value of $\tau$.
By Lemma \ref{lem:potential}, when the algorithm terminates, we have
$\left|W\right|\geq\tau\geq\frac{OPT}{2}-\left(2\log n+1\right)\cdot\left|W\right|-2$,
and therefore in particular $\left|W\right|\geq\frac{OPT-4}{4\log n+4}$. 

Finally, recall that sometimes the algorithm ``hallucinates'' good
realizations, i.e. for some items $i\in W$ that we select, $X_{i}=0$.
However, each time we add an item, the probability that we add a zero-value
item is at most $3/4$ (by the condition $\Pr\left[A\right]>1/4$).
Therefore in expectation the value of the algorithm is at least $\left|W\right|/4$.
\end{proof}

\subsection*{Acknowledgements}
The second author thanks Anupam Gupta for introducing him to  submodular optimization.

\ifFULL
\bibliographystyle{alpha}
\bibliography{prophet2}

\appendix

\section{Missing Proofs}\label{sec:missingProofs}
\noindent \textbf{Lemma~\ref{lem:low-high}}.
Consider any submodular function $f$  and $L,H \in [0,1]$. Let $S^*$ be a set that maximizes $f$, and let $\x \in [0,1]^n$ such that for all $i \in [n]$, $L \leq x_i \leq H$.
Then, \[F(\x) \geq (1-H)(1-L) f(\emptyset) + (1-H)L\cdot f(S^*)\]

\begin{proof}
Assume, wlog, that $S^* = \{1,\dots,k\}$. By submodularity, at each step after we add another element, the potential marginal gain of all other elements decreases. In particular, if we add the elements in $S^*$ in any order, they all have non-negative marginal contribution (since each has a non-negative marginal contribution when added last). 

Let $\x_{\leq i}$ denote the restriction of $\x$ to $[i]$. We first show by induction that for every $i \leq k$, $F(\x_{\leq i}) \geq (1-L)f(\emptyset) + L\cdot f([i])$. Denote $S_i \triangleq S \cap [i]$.
We have that $F(\x_{\leq i})$ is at least:
\begin{align*}   \E_{S \sim \x} \Big[f(S_i)\Big]  & \geq \E_{S \sim \x} \Big[f(S_{i-1}) + f(S_i \cup [i-1]) - f([i-1]) \Big] && \text{(Submodularity)}\\
& = F(\x_{\leq i-1}) + \E_{S \sim \x} \Big[f(S_i \cup [i-1]) - f([i-1]) \Big] \\
& \geq  F(\x_{\leq i-1}) + x_i \Big(f([i]) - f([i-1])\Big) && \text{(Submodularity)} \\
& \geq  F(\x_{\leq i-1}) + L \Big(f([i]) - f([i-1])\Big)  && (f([i]) - f([i-1]) \geq 0).
\end{align*}
Finally by the induction hypothesis, $F(\x_{\leq i-1}) \geq (1-L)f(\emptyset) + L\cdot f([i-1])$.
In particular, we now have that 
\[F(\x_{\leq k}) \geq (1-L) f(\emptyset) + L\cdot f(S^*).\]

It is left to argue that the rest of the elements do not hurt the value too much.  
Consider any $S_i \subseteq S^*$, and let $B^* = B^*(S_k) = \{k+1, \dots ,\ell \}$ be the worst set that we could add to $S_k$, i.e. the $B$ that minimizes $f(S_k \cup B)$.
Let $B_j \triangleq \{k+1, \dots ,\ell \} \subseteq B^*$, and let $T_j$ denote the intersection of $B_j$ with a set $T \subseteq T^*$ sampled according to $\x$.
Now consider two options for adding elements from $T^*$ to $S_k$: 
\begin{enumerate}
\item deterministically, or 
\item independently at random with probabilities sampled according to $\x$. 
\end{enumerate}
Since $f$ is non-negative, we have that even when we add all the bad elements deterministically, 
\begin{gather}\label{eq:from-positivity}
 \sum f(S_k \cup B_j) - f(S_k \cup B_{j-1})   = f(S_k \cup B^*) - f(S_k) \geq - f(S_k).
\end{gather}
When we add the elements at random, we have (by submodularity) that the marginal contribution of each bad element can only increase compared to its contribution in the first case. 
Therefore,
\begin{align*}
 \E_T \left[f(S_k \cup T)  -f(S_k) \right] & = \sum \E_T\left[ f(S_k \cup T_j) - f(S_k \cup T_{j-1})\right] \\
& \geq \sum x_j \left(f(S_k \cup B_j) - f(S_k \cup B_{j-1})\right)  && \text{(Submodularity)}\\
& \geq \sum H \left(f(S_k \cup B_j) - f(S_k \cup B_{j-1})\right)  && (f(S_k \cup B_j) - f(S_k \cup B_{j-1}) \leq 0) \\
& \geq - H f(S_k)  && \text{(Inequality \eqref{eq:from-positivity})}.
\end{align*}

 So far we have $\E_T \left[f(S_k \cup T)\right] \geq (1-H) f(S_k)$. Finally, the marginal contribution of the remaining elments $\{\ell+1,\dots, n\}$ is non-negative by submodularity (if it were negative, we could get a worse set $B'$).
Therefore, for every $S_k$, adding the rest of the elements can decrease the value by at most a factor of $1-H$.
\end{proof}

\end{document}
\fi